\documentclass{sig-alternate}
\usepackage[latin1]{inputenc}
\usepackage{url}
\usepackage{stmaryrd}
\usepackage{float}
\usepackage{pgf}
\usepackage{subfig}
\usepackage{multirow}
\usepackage{xspace}
\usepackage{hyperref}
\usepackage[noend]{algorithmic}
\usepackage{algorithm}
\usepackage{amsmath}
\usepackage{calc} 

\usepackage{color}
\usepackage{bbm}
\usepackage{xspace}
 
\newtheorem{theorem}{Theorem}
\newtheorem{definition}[theorem]{Definition}
\newtheorem{lemma}[theorem]{Lemma}
\newtheorem{lemma*}{Lemma}
{}
{}

\newtheorem{proposition}[theorem]{Proposition}

\newcommand {\C}   {\mathbb C}

\newcommand {\D}   {\mathbb D}

\newcommand {\Z}   {\mathbb Z}

\newcommand {\Q}   {\mathbb Q}

\renewcommand {\P}   {\mathbb P}

\newcommand{\calT}   {\mathcal T}

\newcommand{\OO}{\ensuremath{{{O}}}}

\newcommand{\sO}{\ensuremath{\widetilde{{O}}}}
\newcommand{\sOB}{\ensuremath{\widetilde{{O}}_B}}

\newcommand{\makeremark}[2]{
  \newcommand{#1}[1]
    {
    \color{blue}
     $\longrightarrow$ \textsc{#2: }
     ##1
     $\longleftarrow$
    \color{black}
    }
}    

\newcommand{\blue}[1]{\color{blue}#1\color{black}\xspace}

\renewcommand{\blue}[1]{#1\xspace}

\makeremark{\SL}{SL}
\makeremark{\FR}{Fabrice says}
\makeremark{\MP}{Marc says}
\makeremark{\ET}{Elias says}
\makeremark{\LP}{Luis shouts}
\makeremark{\YB}{Yacine says}

\definecolor{1ST}{rgb}{1,0,0}
\definecolor{2ND}{rgb}{1,0.5,0}
\definecolor{3RD}{rgb}{1,0,1}

\newcommand{\shutup}[1]{}

\renewcommand{\leq}{\leqslant}  
\renewcommand{\geq}{\geqslant}

\makeatletter
\def\cramped                           
 {\parskip0pt\@topsep0pt       
  \itemsep0pt\parsep0pt
}        
\makeatother

\usepackage{stmaryrd}

\usepackage{verbatim}

\title{Improved  algorithm for computing separating linear forms for bivariate systems}
\numberofauthors{5}
\author{
\alignauthor
Yacine Bouzidi\\
       \affaddr{INRIA Nancy Grand Est}\\
       \affaddr{LORIA, Nancy, France}\\
        \email{Yacine.Bouzidi@inria.fr}
\alignauthor
Sylvain Lazard\\
      \affaddr{INRIA Nancy Grand Est}\\
       \affaddr{LORIA, Nancy, France}\\
        \email{Sylvain.Lazard@inria.fr}
\alignauthor 
Guillaume Moroz\\
      \affaddr{INRIA Nancy Grand Est}\\
       \affaddr{LORIA, Nancy, France}\\
         \email{Guillaume.Moroz@inria.fr}
\and
\alignauthor 
Marc Pouget\\
      \affaddr{INRIA Nancy Grand Est}\\
       \affaddr{LORIA, Nancy, France}\\
         \email{Marc.Pouget@inria.fr}
\alignauthor 
Fabrice Rouillier\\
       \affaddr{INRIA Paris-Rocquencourt  IMJ, 
Paris, France}\\
       \email{Fabrice.Rouillier@inria.fr}
}

\begin{document}
\conferenceinfo{ISSAC'14}{July 23--25, 2014, Kobe, Japan.}
\CopyrightYear{2014}
\crdata{978-1-4503-2501-1/14/07}
\maketitle
\begin{abstract}
We address the problem of computing a  linear separating form of  a system of two bivariate
polynomials with integer coefficients, that is a linear combination of the variables that takes
different values when evaluated at the distinct  solutions of the system.
The computation of such linear forms is at the core of most
  algorithms that solve algebraic systems by computing rational
  parameterizations of the solutions and this  is the bottleneck of these algorithms in terms of worst-case bit
  complexity.
We present for this problem a new algorithm of worst-case bit complexity $\sOB(d^7+d^6\tau)$
where $d$ and $\tau$ denote respectively the maximum degree and bitsize of the input
(and where $\sO$ refers to the complexity where
  polylogarithmic factors are omitted and $O_B$ refers to the bit complexity).
This
algorithm simplifies and  decreases by a factor $d$ the worst-case
bit complexity presented for this problem  by Bouzidi et al.
\blue{\cite{bouzidiJSC2014a}}.
This algorithm also yields, for this problem, a probabilistic Las-Vegas algorithm  of expected bit complexity  $\sOB(d^5+d^4\tau)$.
\end{abstract}



\section{Introduction}\label{sec:intro}

A classical approach for solving a system of
polynomials with a finite number of solutions is to compute 
a rational parameterization of its solutions. 

A rational parameterization is a representation of the (complex) solutions 
by a set of univariate polynomials and associated rational one-to-one mappings that send
the roots of the univariate polynomials to the solutions of the system.
Such representations enable to reduce computations on the system to
computations with univariate polynomials and thus ease, for instance, the
isolation of the solutions or the evaluation of other polynomials at the
solutions.

At the core of the algorithms that compute such 
paramete\-rizations (see for example
\cite{ABRW,bostan2003fast,det-jsc-2009,GLS01,VegKah:curve2d:96,Rou99} and references therein), is
the computation of a so-called \emph{linear separating form} for the solutions, that is  a linear combination of the coordinates that takes different values when evaluated at different solutions of the system.
Since a random linear form is
separating with probability one, 
probabilist Monte-Carlo algorithms can overlook this issue.
However, when it comes to deterministically computing a linear separating
  form, or even to check that an arbitrary chosen form is separating, this, surprisingly, turns out to be the bottleneck in the computation of rational parameterizations, in particular for bivariate systems as discussed below. This explains why, among the many algorithms that compute rational parameterizations, seldom search deterministically for a separating linear form. 

Considering systems of two bivariate polynomials of
  total degree bounded by $d$ with integer coefficients of bitsize bounded by
  $\tau$, one approach for computing a separating linear form together with a rational parameterization of the solutions has been presented by Gonzalez-Vega and El Kahoui \cite{VegKah:curve2d:96} and its bit complexity analyzed in~\cite{det-jsc-2009}. The analysis of this approach shows a bit complexity in $\sOB(d^{10}+d^9\tau)$ for computing a separating form and a bit complexity in $\sOB(d^{7}+d^6\tau)$ for computing the corresponding rational parameterization. The computation of a separating linear form was thus the bottleneck in the computation of the rational parameterization.  This is still true even
  when considering the additional phase of computing isolating boxes of the
  solutions (from the rational parameterization), which state-of-the-art
  complexity is in 
  $\sOB(d^8+d^7\tau)$ 
\blue{\cite[Prop. 35]{bouzidiJSC2014a}}.


More recently, Bouzidi et al.~\blue{\cite{bouzidiJSC2014a}} presented a new
algorithm for computing a separating linear form that reduces the previous bit complexity to
$\sOB(d^8+d^7\tau)$.
The same authors also showed that, given such a
separating linear form, an
alternative rational parameterization called RUR \cite{Rou99} can be computed using
$\sOB(d^7+d^6\tau)$ bit operations~\blue{\cite[Thm. 22]{bouzidiJSC2014a}}
and that isolating boxes
of the solutions can be computed from this RUR in $\sOB(d^6+d^5\tau)$ \cite[Thm. 6.1.2]{bouzidi:thesis}.
Consequently, despite the complexity improvement brought to the separating form computation, this step was still the bottleneck in the computation of a rational parameterization of a bivariate system and more generally in the whole solving process, i.e. including the numerical isolation phase.

In addition, although the problem of searching deterministically for a separating form is interesting from the theoretical point of view, in practice, a preferable approach would be to design a Las-Vegas algorithm that chooses randomly  a linear form and then checks that the latter is separating. However up to now, the problem of checking that an arbitrary linear form is separating has not been shown to be easier (at least in terms of asymptotic bit complexity) than the deterministic computation of a separating linear form.

\smallskip

\noindent{\bf Main results.}\quad
Our main contribution is a new deterministic algorithm of worst-case bit
complexity $\sOB(d^7+d^6\tau)$ for computing a separating linear form of a zero-dimensional system of two
bivariate polynomials of total degree at most $d$ and integer coefficients of
bitsize at most $\tau$ (Theorem~\ref{thm:final}). 

This algorithm is simpler than the one presented by Bouzidi et al.~\blue{\cite{bouzidiJSC2014a}} 
and it
decreases by a factor $d$ its complexity. This 
brings the complexity of solving bivariate systems by computing a rational parameterization to $\sOB(d^7+d^6\tau)$. 

A second contribution is a Las-Vegas algorithm for computing a separating linear form with an expected bit complexity in $\sOB(d^5+d^4\tau)$ (Theorem~\ref{thm:las-vegas}). This Las-Vegas algorithm stems naturally from the previous algorithm replacing the deterministic version of the univariate gcd computation by a Las-Vegas one.
Recall that, \blue{in Las-Vegas algorithms, the result is always correct and only the running time is probabilistic.}

 \section{Overview}
 \label{sec:overview}

Our algorithm is based on the one presented by Bouzidi et
al.  
 \blue{\cite{bouzidiJSC2014a}} on the same problem. For clarity, we briefly
recall the essence of that algorithm. It  first computes the number of distinct (complex) solutions of the input system $\{P,Q\}$ as well as a
prime number $\mu$ such that the input system considered modulo $\mu$ has the same number of
distinct solutions. This first step has worst-case bit complexity $\sOB(d^8+d^7\tau)$. 
All polynomials and computations are thereafter considered modulo
$\mu$. The algorithm then considers iteratively
a candidate separating element $x+ay$ with an integer $a$  incrementing  from 0. 
The input polynomials are considered through a  \blue{shearing} of
the coordinate system $(x,y)\leadsto(t-ay,y)$, and the degree of the squarefree
part of their resultant (with respect to $y$) is computed; 
in other words, the algorithm computes the number of
distinct solutions after  projection along the direction of the line $x+ay=0$. The algorithm stops
when a value $a$ is found such that the number of distinct projected solutions equals that of the
system. This step trivially computes a separating element $x+ay$ of the input system considered modulo
$\mu$ but the proof that this form is also separating of the  input system is not straightforward. 
This second step of the algorithm is presented in 
\blue{\cite{bouzidiJSC2014a}} with the same 
worst-case bit complexity as the first step  but we show \blue{in Section~\ref{sec:sep-form}} that it is straightforward to slightly  modify it so
that it has complexity \blue{$\sOB(d^7+d^3\tau)$.}

We present in this paper an improvement of the above algorithm using the following two
ingredients. First, we
show \blue{in Section~\ref{sec:PQ_curve}}  that computing a separating linear form for a system $\{P,Q\}$ is essentially equivalent (in terms of
asymptotic bit complexity) to computing a separating linear form for the critical points of a curve.
Second, 
we present \blue{in Section~\ref{sec:curve}}  a rather simple algorithm of worst-case bit complexity $\sOB(d^7+d^6\tau)$ for computing the number of critical points of a
curve, as well as a prime number $\mu$ such that the curve modulo $\mu$ has the same number of
 critical points. In essence, given a curve of equation $H$, this algorithm first computes a subresultant-based
triangular decomposition \cite{VegKah:curve2d:96} of the
system $\{H,\frac{\partial H}{\partial y}\}$ and the sum of the degrees of the resulting systems; the same
computation is done for the system $\{H,(\frac{\partial H}{\partial y})^2\}$ and we show that the
difference of these two sums of degrees is equal to the number of critical points of the 
curve $H$. We then perform the same computation modulo some prime numbers $\mu$ until the same
number of critical points is obtained. Finally, given this number of solutions and a
corresponding prime $\mu$, we obtain a separating linear form for the input system by  applying  the
\blue{variant presented in Section~\ref{sec:sep-form} of the} algorithm of 
\blue{\cite{bouzidiJSC2014a}} for computing a separating linear form for the critical points of the
curve.


\blue{Furthermore, we show in Section~\ref{sec:Las-Vegas} how this algorithm naturally extends to a Las-Vegas algorithm of
  expected bit complexity $\sOB(d^5+d^4\tau)$.}


\section{Notation and preliminaries}\label{sec:prelim}

We introduce notation and recall some classical material. Most of the material in this section is
taken literally from~\blue{\cite{bouzidiJSC2014a}}. 

 The bitsize of an integer $p$ is the number of bits needed to represent it, that
is $\lfloor\log p\rfloor+1$ ($\log$ refers to the logarithm in base 2). 
The bitsize of a polynomial with integer
coefficients is the \emph{maximum} bitsize of its coefficients. As mentioned
earlier, $O_B$ refers to the bit complexity and $\sO$ and $\sOB$ refer to
complexities where polylogarithmic factors are omitted, see \cite[Def.
25.8]{vzGGer2} for details. 


In the following, $\mu$ is a prime number and we denote by $\Z_\mu$ the quotient
$\Z/\mu\Z$.  We denote by $\phi_\mu$: $\Z \rightarrow \Z_\mu$ the reduction
modulo $\mu$, and extend this definition to the reduction of polynomials with
integer coefficients.  We denote by $\D$ a unique factorization domain,
typically $\Z[x,y]$, $\Z[x]$, $\Z_\mu[x]$, $\Z$ or $\Z_\mu$. We also denote by
 $\mathbb{F}$ a field, typically $\Q$, $\C$, or $\Z_\mu$ and by
 $\mathbb{F}_{\D}$ the fraction field of $\D$.

For any polynomial $P\in \D[x]$, let $Lc_x(P)$ denote its leading coefficient
with respect to the variable $x$ and
$d_x(P)$ its degree with respect to $x$.
For any curve defined by $H(x,y) \in \D[x,y]$, we call the critical points of $H$ with respect to $x$ or more shortly the critical point of $H$, the points that are solutions of the system $\{H,\frac{\partial H}{\partial y}\}$.
In this paper, 
the solutions of a system of polynomial are always considered in the algebraic closure of
$\mathbb{F}_{\D}$.

\smallskip
\noindent{\bf Subresultant sequences.}\quad
We first recall the concept of
\emph{polynomial determinant} of a matrix which is used in the definition of subresultants.
Let $M$ be an $m \times n$ matrix with $m \leq n$ and $M_i$ be the square submatrix of $M$ consisting of the first $m-1$ columns
and the $i$-th column of $M$, for $i=m,\ldots,n$. The \emph{polynomial determinant} of $M$ is the polynomial
defined as $\det(M_{m})y^{n-m}+\det(M_{m + 1})y^{n-(m+1)}+\cdots + \det(M_{n})$.

Let $P =\sum_{i=0}^p a_i y^i$ and $Q= \sum_{i = 0}^q b_i y^i$ be two polynomials
in $\D[y]$ and assume without loss of generality that $p \geq q$. The Sylvester
matrix of $P$ and $Q$, $Sylv(P,Q)$ is the $(p+q)$-square matrix whose rows are
$y^{q-1}P,\ldots,P,y^{p-1}Q,\ldots,Q$ considered as vectors in the basis
$y^{p+q-1},\ldots,y,1$.

\begin{definition}\label{def:sub-resultant}(\cite[\S 3]{Kahoui03}). 
  For $i=0,\ldots, \min(q,p-1)$, let $Sylv_i(P,Q)$ be the $(p+q-2i)\times
  (p+q-i)$ matrix obtained from $Sylv(P,Q)$ by deleting the $i$ last rows of the
  coefficients of $P$, the $i$ last rows of the coefficients of $Q$, and the $i$
  last columns.
 
  For $i=0,\ldots, \min(q,p-1)$, the $i$-th polynomial subresultant of $P$ and
  $Q$, denoted by $Sres_{y,i}(P,Q)$ is the polynomial determinant of
  $Sylv_i(P,Q)$.  When $q=p$, the $q$-th polynomial subresultant of $P$ and $Q$
  is $b_q^{-1}Q$.%
\end{definition}

$Sres_{y,i}(P,Q)$ has degree at most $i$ in $y$, and the coefficient of its
monomial of degree $i$ in $y$, denoted by ${sres}_{y,i} (P, Q)$, is called the
$i$-th \emph{principal subresultant coefficient}.  Note that ${Sres}_{y,0} (P,
Q)={sres}_{y,0} (P, Q)$ is the \emph{resultant} of $P$ and $Q$ with respect to
$y$, which we also denote by $Res_y(P,Q)$.

\smallskip
We state below a fundamental property of subresultants which is instrumental in
the triangular decomposition algorithm used in Section~\ref{sec:nb-critical}. For
clarity, we state this property for bivariate polynomials $P =\sum_{i=0}^p a_i
y^i$ and $Q= \sum_{i = 0}^q b_i y^i$ in $\D[x,y]$, with $p\geq q$.
Note that this property is often stated with a stronger assumption that is that
\emph{none} of the leading terms $a_p(\alpha)$ and $b_q(\alpha)$ vanishes.  This
property is a direct consequence of the specialization property of subresultants
and of the gap structure theorem; see for instance \cite[Lemmas 2.3, 3.1 and
Cor. 5.1]{Kahoui03}.

\begin{lemma}\label{lem:fund-prop-subres}
  For any $\alpha$ such that $a_p(\alpha)$ and $b_q(\alpha)$ do not both vanish,
  the first ${Sres}_{y,k}(P,Q)(\alpha,y)$ (for $k$ increasing) that does not
  identically vanish is of degree $k$ and it is the gcd of $P(\alpha,y)$ and
  $Q(\alpha,y)$ (up to a nonzero constant in the fraction field of
  $\D(\alpha)$).
\end{lemma}

\noindent{\bf Complexity.}\quad
We recall complexity results, using fast algorithms, on subresultants and
gcd computations. 

\begin{lemma}[{\cite[Prop. 8.46]{BPR06} \cite[\S 8]{Reischert1997} \cite[Cor. 11.15]{vzGGer2}}]
  \label{complexity:subresultant}
  Let $P$ and $Q$ be in $\mathbb{Z}[x_1,\ldots, x_n][y] $ ($n$ fixed) with
  coefficients of bitsize at most  $\tau$
such that their degrees in $y$ are bounded by
  $d_{y}$ and their degrees in the other variables are bounded by~$d$. 
  \begin{itemize}\cramped
  \item The coefficients of $Sres_{y,i}(P,Q)$ have bitsize in
    $\sO(d_{y}\tau)$. 
  \item The degree in $x_j$ of $Sres_{y,i}(P,Q)$ is at most
    $2d(d_{y}-i)$. 
  \item Any subresultant $Sres_{y,i}(P,Q)$ as well as the sequence of principal subresultant coefficients $sres_{y,i}(P,Q)$ can be computed in
$\sO(d^{n} d_{y}^{n+1})$ arithmetic operations, 
    and $\sOB(d^{n}
    d_{y}^{n+2}\tau)$ bit operations. 
\end{itemize}
\end{lemma}    

In the sequel, we often consider the gcd of two univariate polynomials $P$ and
$Q$ and the gcd-free part of $P$ with respect to $Q$, that is, the divisor $D$
of $P$ such that $P=\gcd(P,Q)D$. Note that, when $Q=P'$,  the latter is the
squarefree part of~$P$, \blue{provided that the characteristic of the
  coefficient ring is zero or sufficiently large (e.g., larger than the degree of~$P$)}. 

\begin{lemma}[{\cite[Rem. 10.19]{BPR06}}]
\label{complexity:gcd}
  Let $P$ and $Q$ in $\mathbb{F}[x]$ of degree at most $d$. $\gcd(P,Q)$
or the gcd-free part of $P$ with respect to $Q$ can be computed with
  $\sO(d)$ operations in~$\mathbb{F}$.
\end{lemma}

\section{Separating linear form}
\label{sec:sep-form}

As mentioned in the overview, our approach for computing a separating form of a zero-dimensional
system $\{P,Q\}$   is similar to the one
in \blue{\cite{bouzidiJSC2014a}} \blue{once we know} the number of distinct solutions and a
so-called lucky prime $\mu$. \blue{Such a  lucky prime} is, roughly speaking, a prime such
that $\{P,Q\}$ \blue{has} the same number of distinct solutions as its 
\blue{image modulo $\mu$.}
Before presenting Algorithm \ref{alg:sep-elem}, \blue{which computes a separating linear form in
this context,} we introduce  the following notation and  formally define lucky primes.

%

Given the two input polynomials $P$ and $Q$, 
we  consider the ``generic'' change of variables $x=t-sy$, and define the
 ``sheared'' polynomials $P(t-sy,y)$, $Q(t-sy,y)$,
 and their resultant with respect to $y$, 
\[{ R(t,s)}=Res_y({P}(t-sy,y),{Q}(t-sy,y)).\]
 We introduce the following notation for the leading coefficients of these polynomials; 
\[L_{P}(s) = Lc_y({P}(t-sy,y))\quad L_{Q}(s) = Lc_y({Q}(t-sy,y)).\]
Note that these polynomials do not depend on $t$.

\begin{definition}[{\cite[Def. 8]{bouzidiJSC2014a}}]\label{def:lucky-nb-roots}
 A prime number $\mu$ is said to be {\em \bf lucky} for a zero-dimensional system  $\{P,Q\}$ if
 $\{P,Q\}$ and $\{\phi_\mu(P),\phi_\mu(Q)\}$ have the same  number of distinct 
 solutions and if $\mu>2d^4$   and
\[\phi_\mu(L_P(s))\ \phi_\mu(L_Q(s))\not\equiv 0.\]
\end{definition}

\blue{Note that we consider $\mu$ in $\Omega(d^4)$ in Definition~\ref{def:lucky-nb-roots}  because,
in Algorithm~\ref{alg:sep-elem},  we want to ensure
that there exists, for the system $\{P_\mu,Q_\mu\}$ (resp. $\{P,Q\}$), a separating form $X+aY$ {with $a\in\Z_\mu$}
(resp. $0\leq a<\mu$ in $\Z$). The constant 2 in
the bound $2d^4$ is an overestimate, which simplifies some  proofs in \blue{\cite{bouzidiJSC2014a}}.}

Recall that we consider we know the number of distinct (complex) solutions of  system $\{P,Q\}$ and a
lucky prime $\mu$ for that system. 
Algorithm~4 
of~\blue{\cite{bouzidiJSC2014a}} 
computes a separating linear form for $\{P,Q\}$ by
considering iteratively linear forms $x+ay$, where $a$ is an integer incrementing from 0
and by  computing the degree of the squarefree part of the reduction modulo $\mu$ of $R(t,a)$ until
this degree is equal to the (known) number of distinct solutions of the system and such that
$\phi_\mu(L_P(a))\ \phi_\mu(L_Q(a))\neq 0$. 

Doing so, the algorithm computes a separating form for the system modulo $\mu$,
which, under the hypothesis of the luckiness of $\mu$, has been proven to be
also separating for the system $\{P,Q\}$.  In Algorithm \ref{alg:sep-elem}, we
follow the same approach except that we perform the computations in a slightly
different way\footnote{\small \blue{Namely, in Algorithm \ref{alg:sep-elem}, we
    first compute the reduction modulo $\mu$ of the input polynomials $P$ and
    $Q$ (Line~\ref{alg:sep3}) and then, for every value of $a$, the resultant of
    their sheared images through the change of variables $(x,y)\leadsto(t-ay,y)$
    (Line~\ref{alg:sep6}), while in \cite[Algorithm~4]{bouzidiJSC2014a}, we
    first compute the reduction modulo $\mu$ of the resultant $R(t,s)$ and then,
    for every value of $a$, its specialization at $s=a$.}} so that the
complexity is in $\sOB(d^7+d^3{\tau})$ (instead of $\sOB(d^8+d^7{\tau})$ in
\blue{\cite{bouzidiJSC2014a}}).

\begin{algorithm}[t]
  \caption{Separating form for $\{P,Q\}$}
\label{alg:sep-elem}
\begin{algorithmic}[1]
  \REQUIRE{ $P,Q$ in $\Z[x,y]$ of total degree at most $d$ and defining a zero-dimensional
system,  its number $N$ of  distinct (complex) solutions and a lucky prime $\mu$ of bitsize
  \blue{$\OO(\log d)$}}
  \ENSURE{A separating linear form $x+ay$ for $\{P,Q\}$, with $a<2d^4$}

\STATE Compute $P(t-sy,y)$ and  $Q(t-sy,y)$
\label{alg:sep1}
\STATE Compute $\Upsilon_\mu(s)=\phi_\mu(L_P(s))\ \phi_\mu(L_Q(s))$\label{alg:sep2}
\STATE Compute $P_\mu=\phi_\mu(P)$ and $Q_\mu=\phi_\mu(Q)$\label{alg:sep3}
\STATE $a:=0$\label{alg:sep4}
\REPEAT \label{alg:sep5}
\STATE Compute $P_\mu(t-ay,y)$, $Q_\mu(t-ay,y)$ and their resultant
$R_{\mu,a}(t)$\label{alg:sep6}
\STATE \label{alg:sep7}Compute the degree $N_a$ of the squarefree part of $R_{\mu,a}(t)$
\STATE $a:=a+1$ 
\UNTIL $\Upsilon_\mu(a)\neq 0$\footnotemark and $N_a=N$ \label{alg:sep9}
\RETURN The linear form $x+ay$\label{alg:sep10}
\end{algorithmic}
\end{algorithm}
\footnotetext{\small $\Upsilon_\mu(s)\in\Z_\mu[s]$ and we consider  $\Upsilon_\mu(a)$ in $\Z_\mu$.}

\begin{proposition}\label{prop:sep-elem-comp}
Algorithm~\ref{alg:sep-elem}  computes a separating linear form $x+ay$ for $\{P,Q\}$ with
$a<2d^4$ with a bit complexity  $\sOB(d^7+d^3{\tau})$. 
\end{proposition}

\begin{proof}
We first prove the correctness of Algorithm~\ref{alg:sep-elem}  which essentially follows from 
\blue{\cite[Algorithm~4]{bouzidiJSC2014a}}. 
The latter algorithm computes the degree of the
squarefree part of $\phi_\mu(R)(t,a)$ until the condition of Line~\ref{alg:sep9} is satisfied, and
it returns the corresponding form $x+ay$.
It is thus sufficient to argue that $\phi_\mu(R)(t,a)=R_{\mu,a}(t)$. 

Denoting by $\psi_a$ the morphism that evaluates a polynomial at $s=a$, and $\text{Res}_y$ the
resultant with respect to $y$, we have 
\[\begin{split}
\phi_\mu(R)(t,a)&=\psi_a\circ\phi_\mu(\text{Res}_y(P(t-sy,y),Q(t-sy,y)) =\\
&\text{Res}_y(\psi_a\circ\phi_\mu(P(t-sy,y)), \psi_a\circ\phi_\mu(Q(t-sy,y)))
\end{split}\]
by the specialization property of the resultants 
since the leading coefficients of $P$ and $Q$ (with respect to $y$) do not vanish through
$\psi_a\circ\phi_\mu$ when the condition $\Upsilon_\mu(a)\neq 0$ is
satisfied in  Line~\ref{alg:sep9}. Furthermore, $\phi_\mu(P(t-sy,y))=\phi_\mu(P)(t-sy,y)$ and similarly for $Q$, which
implies that the right-hand side  of the equation is equal to $R_{\mu,a}(t)$. This concludes the proof
of correctness. Note that this correctness includes the property that the output
integer $a$ is less than $2d^4$.

We now prove the complexity of our algorithm. 
It is
straightforward that,  in Line~\ref{alg:sep1}, the sheared polynomials $P(t-sy,y)$ and
$Q(t-sy,y)$ can be computed  in bit complexity
$\sOB(d^4+d^3\tau)$ and that their bitsizes are in $\sO(d+\tau)$ (see e.g. \blue{\cite[Lemma 7]{bouzidiJSC2014a}}). 
In Lines \ref{alg:sep2} and \ref{alg:sep3}, the polynomials, in one or two variables, have degree at most $d$ and bitsize $\sO(d+\tau)$. The reduction of each of their $O(d^2)$ coefficients modulo $\mu$ can be done in a bit complexity that is softly linear in the
 maximum bitsizes \cite[Thm. 9.8]{vzGGer2}, that is  in a total bit complexity of $\sOB(d^3+d^2\tau)$. 
In Line  \ref{alg:sep6}, computing the polynomials $P_\mu(t-ay,y)$ and $Q_\mu(t-ay,y)$ is performed using
$\sOB(d^3)$ bit operations (see e.g. the proof \blue{\cite[Lemma 7]{bouzidiJSC2014a}}) 
and similarly for their resultant $R_{\mu,a}(t)$ according to
Lemma~\ref{complexity:subresultant}. 
In Line  \ref{alg:sep7}, the squarefree part of $R_{\mu,a}(t)$ can also be computed in $\sOB(d^3)$
bit operations by Lemma~\ref{complexity:gcd}, since the resultant has degree $O(d^2)$.
We have shown that the loop stops with $a<2d^4$, thus the whole  loop has complexity
$\sOB(d^7)$, which concludes the proof.
\end{proof}

\section{From a system to a curve}\label{sec:PQ_curve}

In this section, we consider two polynomials $P,Q \in \Z[x,y]$ of total degree at most $d$ and maximum bitsize $\tau$ and show that it is essentially equivalent from an asymptotic worst-case bit
complexity point of view to compute a separating linear form for a system $\{P,Q\}$ and  to compute a
separating linear form for the critical points of a curve.  For simplicity, we refer to the latter
as a separating linear form for a curve.  

By definition, the critical points of a
curve of equation $H$ are the solutions of the system $\{H,\frac{\partial H}{\partial y}\}$, thus
computing a separating linear form for a curve amounts by definition to computing a
separating linear form for a system of two equations. Conversely, a separating linear form for the curve $PQ$ is also separating for
the system $\{P,Q\}$ since any solution of $\{P,Q\}$ is also solution of $PQ$ and of
$\frac{\partial PQ}{\partial y}=P\frac{\partial Q}{\partial y}+ \frac{\partial P}{\partial y}
Q$. 

However, it may happen that the curve $PQ$ admits no separating linear form even
if $\{P,Q\}$ admits one. Indeed, $\{P,Q\}$ can be zero-dimensional while $PQ$ is
not squarefree (and such that the infinitely many critical points cannot be
separated by a linear form).  Nevertheless, if $P$ and $Q$ are coprime and
squarefree, then $PQ$ is squarefree and thus it has finitely many singular
points.
\blue{Still the curve $H=PQ$ may contain vertical lines, and thus infinitely many
critical points, but this issue can easily be handled by shearing the coordinate
system.}

\begin{lemma}
\label{lem:system2curve}
Given a zero-dimensional system of two polynomials $P$ and $Q$ in $\Z[x,y]$ of maximum degree $d$ and maximum
bitsize $\tau$, we can compute in complexity $\sOB(d^6+d^5\tau)$ a \blue{shearing} of the coordinate system
$(x,y)\leadsto(t-\alpha y,y)$ ($\alpha$ integer in $O(d)$) and a polynomial $H$ in $\Z[t,y]$ of degree at most $2d$ and bitsize $\sO(d+\tau)$  so
that the system $\{H,\frac{\partial H}{\partial y}\}$ is zero-dimensional and any separating linear
form for that system is also separating for $\{P,Q\}$ after being sheared back. 
\end{lemma}
\begin{proof}
As discussed above, we  first compute the squarefree part of \blue{each polynomial $P$ and $Q$}, which can be
done in complexity $\sOB(d^6+d^5\tau)$ \cite[Lemma 13]{sagraloff2013}. \blue{Let $H(x,y)$ denote
  their product, which is squarefree since $P$ and $Q$ are coprime.}
We then consider a generic \blue{shearing} of the
coordinate system $(x,y)\leadsto(t-s y,y)$ in order to find a value $s=\alpha$ so that the sheared
curve $\widehat{H}(t,y)=H(t-\alpha y,y)$ has no vertical asymptote \blue{ and thus no
vertical line.}
The leading coefficient of $H(t-s y,y)$ (seen as a polynomial in $y$) is a polynomial of degree at
most $d$ in $\Z[s]$ ($t$ does not appear in the leading
term); furthermore an expanded form of $H(t-s y,y)$  can  be computed in complexity
$\sOB(d^4+d^3\tau)$ and the coefficients have 
bitsize $\sO(d+\tau)$ (see e.g. \blue{\cite[Lemma 7]{bouzidiJSC2014a}}). 
Finding an integer value $s=\alpha$ where the leading coefficient does not vanish can thus be done in $d$
evaluations of complexity $\sOB(d(d+\tau))$ each \blue{\cite[Lemma 6]{bouzidiJSC2014a}} 
and such  $\alpha$ can be found in \blue{ $[0,d]$}. Then, computing $H(t-\alpha y,y)$   can be done by evaluating
each of the coefficients of $H(t-s y,y)$ at $s=\alpha$, which can again be done with $O(d)$
evaluations of  complexity $\sOB(d(d+\tau))$ each. 
Thus, we can shear the curve in complexity $\sOB(d^4+d^3\tau)$ so that the leading coefficient of the resulting polynomial $\widehat{H}(t,y)=H(t-\alpha y,y)$ (seen
as a polynomial in $y$) is a constant.

Modulo the \blue{shearing}, all solutions of $\{P,Q\}$ are solutions of the system $\{\widehat{H},\frac{\partial
\widehat{H}}{\partial y}\}$. Indeed, a solution $(x_0,y_0)$ of $\{P,Q\}$ is such that
$(t_0=x_0+\alpha y_0,  y_0)$ is solution of $\{\widehat{P},\widehat{Q}\}$ with
$\widehat{P}(t,y)$ equal to the squarefree part of $P(t-\alpha y,y)$ and similarly for
$\widehat{Q}$; thus $(t_0,  y_0)$ is solution of $\widehat{H}=\widehat{P}\widehat{Q}$ and of $\frac{\partial
\widehat{H}}{\partial y}=\widehat{P}\frac{\partial
\widehat{Q}}{\partial y} + \frac{\partial
\widehat{P}}{\partial y} \widehat{Q}$. Thus, any separating linear form for  $\{\widehat{H},\frac{\partial
\widehat{H}}{\partial y}\}$ is also separating for $\{P,Q\}$ modulo the \blue{shearing}. Finally, $\{\widehat{H},\frac{\partial
\widehat{H}}{\partial y}\}$ is zero-dimensional since, by construction,
$\widehat{H}$ is squarefree and contains no vertical
line.
Renaming $\widehat{H}$ by $H$, this concludes the proof.
\end{proof}

\section{The case of a curve}
\label{sec:curve}

In this section, we consider an arbitrary curve defined by  $H\in\Z[x,y]$ of degree $d$ and bitsize
$\tau$, with a constant leading
coefficient in $y$, and such that $H$ has a finite number of critical points, i.e., the system
$\{H,\frac{\partial H}{\partial y}\}$ is zero-dimensional. We show in the following that  (i) computing the number of the critical points of
$H$ and (ii) computing a lucky prime for $\{H,\frac{\partial H}{\partial
    y}\}$ (see Definition~\ref{def:lucky-nb-roots}) can be done in a bit complexity in $\sOB(d^7+d^6\tau)$. 
Combined with the results of the previous sections, this will yield that we can compute a separating
linear form for an arbitrary zero-dimensional system $\{P,Q\}$ in the same complexity.

\subsection{Number of critical points}\label{sec:nb-critical}

Our algorithm for computing  the number of (complex) critical points of a curve is based on  
a classical algorithm for computing a triangular decomposition of
a system of two bivariate polynomials. We first recall this algorithm and then show how it can be slightly modified and used to compute the number of critical points of a curve.

\smallskip
\noindent{\bf Triangular decomposition.}\quad
Let $P$ and $Q$ be two polynomials in $\mathbb{F}[x,y]$ of degree at most $d$.
A decomposition  of the system $\{P,Q\}$ using the
subresultant sequence appears in the theory of triangular sets
\cite{Li-modpn-11} and for the computation of the topology of curves
\cite{VegKah:curve2d:96}.

The idea is based on Lemma~\ref{lem:fund-prop-subres} which states that, after
specialization at $x=\alpha$, the first (with respect to increasing $i$) nonzero
subresultant $Sres_{y,i}(P,Q)(\alpha,y)$ is of degree $i$ and is equal to the
gcd of $P(\alpha,y)$ and $Q(\alpha,y)$.  This induces a decomposition 
into triangular subsystems $(\{A_i(x),$ $Sres_{y,i}(P,Q)(x,y)\})$ where a solution
$\alpha$ of $A_i(x)=0$ is such that the system $\{P(\alpha,y), Q(\alpha,y)\}$
admits exactly $i$ roots (counted with multiplicity), which are exactly those of
$Sres_{y,i}(P,Q)(\alpha,y)$.  Furthermore, these triangular subsystems are
regular chains, i.e., the leading coefficient of the bivariate polynomial (seen
in $y$) is coprime with the univariate polynomial.  For clarity and
self-containedness, we recall this decomposition in
Algorithm~\ref{alg:tri-dec-mod}. 
Note that this algorithm performs $\sO(d^4)$ arithmetic operations in $\mathbb{F}$ (see e.g. \blue{\cite[Lemma 15]{bouzidiJSC2014a}}).  
We also state the following properties which directly follow from the algorithm and Lemma~\ref{lem:fund-prop-subres}.

\begin{lemma}[{\cite{VegKah:curve2d:96,Li-modpn-11}}]\label{lem:tridec-correctness}
   Algorithm~\ref{alg:tri-dec-mod} computes a triangular decomposition 
  $\{(A_i(x),$ $B_i(x,y))\}_{i\in\cal I}$ such that 

\begin{itemize}\cramped
  \item the set of solutions of $\{P,Q\}$ is the disjoint union  of the sets of solutions of the
  $\{A_i(x),B_i(x,y)\}$, ${i\in{\cal I}}$
  \item $\prod_{i\in\cal I}A_i$ is  squarefree, 
  \item for any root $\alpha$ of $A_i$,   
      $B_i(\alpha,y)$ is of degree $i$ and is equal to $\gcd(P(\alpha,y),$  $Q(\alpha,y))$. 
\end{itemize}
\end{lemma}

\begin{algorithm}[t]
  \caption{Triangular decomposition \cite{VegKah:curve2d:96,Li-modpn-11}} 
\label{alg:tri-dec-mod}
\begin{algorithmic}[1]
  \REQUIRE{ $P,Q$ in $\mathbb{F}[x,y]$ coprime such that $Lc_y(P)$ and $Lc_y(Q)$ are
    coprime, $d_y(Q)\leq d_y(P)$}
  \ENSURE{Triangular decomp. 
  $\{(A_i(x),B_i(x,y))\}_{i\in\cal I}$ such that the set of solutions of $\{P,Q\})$ is the disjoint union
  of the sets of solutions of $\{A_i(x),B_i(x,y)\}_{i\in\cal I}$}
\STATE Compute the subresultant sequence of $P$ and $Q$ with respect to $y$: $B_i=Sres_{y,i}(P,Q)$
\STATE  $G_0=\text{squarefree part}({Res_y(P,Q)})$ and ${\calT}=\emptyset$ \label{line2-algo1}
\FOR {$i=1$ \TO $d_y(Q)$}
\STATE $G_i=\gcd(G_{i-1},sres_{y,i}(P,Q))$ \label{line4-algo1}

\STATE $A_i=G_{i-1}/G_i$ \label{line5-algo1}
\STATE if $d_x(A_i)>0$, add $(A_i,B_i)$ to ${\calT}$ 
\ENDFOR
\RETURN ${\calT}= \{(A_i(x),B_i(x,y))\}_{i\in{\cal I}}$
\end{algorithmic}
\end{algorithm}

\begin{algorithm}[t]
  \caption{Degree of the triangular decomposition } 
\label{alg:tri-dec-mod2}
\begin{algorithmic}[1]
  \REQUIRE{ $P,Q$ in $\mathbb{F}[x,y]$ coprime such that $Lc_y(P)$ and $Lc_y(Q)$ are
    coprime, $d_y(Q)\leq d_y(P)$}
  \ENSURE{The degree of the triangular decomposition of $\{P,Q\}$}
\STATE Compute the principal subresultant sequence of $P$ and $Q$ with respect to $y$: $sres_{y,i}(P,Q)$
\STATE  $G_0=\text{squarefree part}({Res_y(P,Q)})$ 
\FOR {$i=1$ \TO $d_y(Q)$}
\STATE $G_i=\gcd(G_{i-1},sres_{y,i}(P,Q))$\label{line4-algo2}

\ENDFOR
\RETURN $\sum_{i\in{\cal I}}(\deg(G_{i-1})-\deg(G_i))\,i$
\end{algorithmic}
\end{algorithm}

\smallskip
\noindent{\bf Degree of the triangular decomposition.}\quad
We call the \emph{degree of the triangular decomposition of $\{P,Q\}$,} 
the sum of the degrees of the triangular systems computed
by Algorithm~\ref{alg:tri-dec-mod}, that is,
\[\sum_{i\in\cal I}{\deg_x(A_i(x))\,\deg_y(B_i(x,y))}\] where $\deg_x$ refers to the degree of the
polynomial with respect to $x$ and similarly for $y$.
As we will see below, we only need the degree of the  triangular decomposition of some systems for computing the
number of critical points of $H$. 

We present in Algorithm~\ref{alg:tri-dec-mod2}  a slight variation of the triangular
decomposition algorithm in which we only compute the degree of the decomposition. 
Instead of computing the subresultant sequence $Sres_{y,i}(P,Q)$ of
$P$ and $Q$ as in Algorithm~\ref{alg:tri-dec-mod}, we only compute  the sequence of principal
subresultant coefficients of
$P$ and $Q$ (that is, the sequence of coefficients of the monomials of degree $i$ in $y$ in
$Sres_{y,i}(P,Q)$), which is sufficient for computing the degree of the decomposition. As we will
see, this decreases by a factor $d$  the arithmetic complexity in $\mathbb{F}$ of the algorithm, which is
critical for our global algorithm.\footnote{\small Note that, while  this complexity improvement does not
impact the bit complexity of computing the number of critical points of a curve $H$ over $\Z$, it is
critical when computing a lucky prime for $\{H,\frac{\partial H}{\partial y}\}$ where the number
of critical points is computed for $O(d^4+d^3\tau)$ systems defined over distinct  $\Z_\mu$  (Proposition~\ref{prop:comp-lucky}).}

\begin{lemma}\label{lem:tridec_pair}
Algorithm~\ref{alg:tri-dec-mod2} computes the degree of the triangular decomposition of $\{P,Q\}$.
If $P, Q\in\mathbb{F}[x,y]$ have degree at most $d$, the algorithm performs $\sO(d^3)$ arithmetic
operations in $\mathbb{F}$. If 
$P,Q\in\mathbb{Z}[x,y]$ ($\subset\mathbb{Q}[x,y]$) have degree at most $d$ and bitsize at most $\tau$,
the algorithm performs $\sOB(d^7+d^6\tau)$ bit 
operations in $\mathbb{Z}$. 
\end{lemma}

\begin{proof}
The correctness of Algorithm~\ref{alg:tri-dec-mod2} directly follows from
Lemma~\ref{lem:tridec-correctness}. Concerning the complexity, the resultant and the sequence of the principal subresultant coefficients of $P$ and $Q$ can be computed in
  $\sO(d^3)$ arithmetic operations, and each of these principal subresultants (including the
  resultant) has degree in $\OO(d^2)$, by  Lemma \ref{complexity:subresultant} (note that this lemma is stated for
  the coefficient ring $\Z$, but the arithmetic complexity is the same for any
  field $\mathbb{F}$).
  The algorithm performs at most $d$ gcd computations between these polynomials. 
  The arithmetic complexity of one such gcd computation is softly
  linear in their degrees, that is $\sO(d^2)$
  (Lemma~\ref{complexity:gcd}). Hence the complexity of computing 
  all the gcds is in $\sO(d^3)$. 
The bit complexity over $\mathbb{Z}$ Algorithm~\ref{alg:tri-dec-mod2} is bounded by that of
  Algorithm~\ref{alg:tri-dec-mod} which is in $\sOB(d^7+d^6\tau)$ according to the proof of
  \cite[Thm. 19]{det-jsc-2009}.\footnote{\small Note that this bound is not an obvious overestimate because  known bounds yield a complexity of $\sOB(d^7+d^6\tau)$ for all the gcd computations
  in Line~\ref{line4-algo1} of Algorithm~\ref{alg:tri-dec-mod}, which is the
  same for Line~\ref{line4-algo2} of Algorithm~\ref{alg:tri-dec-mod2}.}
\end{proof}

\begin{lemma}\label{lem:mult} 
The degree of the triangular decomposition of $\{P,Q\}$ is equal to the sum, over all distinct
solutions $(\alpha,\beta)$ of $\{P,Q\}$,  of the multiplicities of $\beta$ in $\gcd(P(\alpha,y),Q(\alpha,y))$. 
\end{lemma}
\begin{proof}
By Lemma~\ref{lem:tridec-correctness}, the sets of solutions of the systems of the triangular decomposition of
Algorithm~\ref{alg:tri-dec-mod} are disjoint and polynomials $A_i$ are squarefree. The degree of the
triangular decomposition of $\{P,Q\}$ is thus 
\[\sum_{i\in\cal I}{\deg_x(A_i(x))\,\deg_y(B_i(x,y))}= 
\sum_{(\alpha,\beta) \in V}\text{mult}(\beta,B_i(\alpha,y)),\]
where $V$ is the set of solutions of $\{P,Q\}$ and $\text{mult}(\beta,B_i(\alpha,y))$ denotes the multiplicity of $\beta$ in $B_i(\alpha,y)$.
The result follows since  $B_i(\alpha,y)= \gcd(P(\alpha,y),Q(\alpha,y))$ by Lemma~\ref{lem:tridec-correctness}.
\end{proof}

\noindent{\bf Number of critical points of $H$.}\quad
Algorithm~\ref{alg:rad-tri-dec} computes the number of critical points of $H$  as the difference between the
degree of the triangular decompositions of the systems $\{H,(\frac{\partial H}{\partial y})^2\}$ and $\{H,\frac{\partial
  H}{\partial y}\}$. We first prove the correctness of this algorithm and then its complexity. 

\begin{algorithm}[t]
  \caption{Number of critical points of $H$}
\label{alg:rad-tri-dec}
\begin{algorithmic}[1]
  \REQUIRE{$H$ in $\mathbb{F}[x,y]$ squarefree such that $Lc_y(H) \in \mathbb{F}$}
  \ENSURE{ The number of critical points of $H$}

\smallskip
\RETURN  Algo \ref{alg:tri-dec-mod2} $(H,(\frac{\partial H}{\partial y})^2)$ - Algo \ref{alg:tri-dec-mod2} $(H,\frac{\partial H}{\partial y})$

\end{algorithmic}
\end{algorithm}

\begin{proposition}\label{prop:proof-correctness}
  Algorithm \ref{alg:rad-tri-dec} computes the number of critical points of $H$.
If $H\in\mathbb{F}[x,y]$ has degree $d$, the algorithm performs $\sO(d^3)$ arithmetic
operations in $\mathbb{F}$. If 
$H\in\mathbb{Z}[x,y]$ ($\subset\mathbb{Q}[x,y]$) has degree $d$ and bitsize $\tau$,
the algorithm performs $\sOB(d^7+d^6\tau)$ bit 
operations in $\mathbb{Z}$. 

\end{proposition}

\begin{proof}
We first prove that for any critical point $(\alpha,\beta)$ of $H$, the multiplicity of $\beta$ in
$\gcd(H(\alpha,y),(\frac{\partial H}{\partial y})^2(\alpha,y))$ is greater by one than the
multiplicity of $\beta$ in  $\gcd(H(\alpha,y),$ $\frac{\partial H}{\partial y}(\alpha,y))$. 
Since $(\alpha,\beta)$ is a critical point of $H$, it is  solution of both the systems $\{H,\frac{\partial H}{\partial y}\}$ and $\{H,(\frac{\partial H}{\partial y})^2\}$. This implies that $\beta$ is a root of both $\gcd(H(\alpha,y),\frac{\partial H}{\partial y}(\alpha,y))$ and $\gcd(H(\alpha,y),(\frac{\partial H}{\partial y})^2(\alpha,y))$. If $m$  is the multiplicity of $\beta$ in $H(\alpha,y)$ then $\beta$ has multiplicity $m-1$ in $\frac{\partial H}{\partial y}(\alpha,y)$ and thus, that it has multiplicity $2m-2$ in $(\frac{\partial H}{\partial y})^2$. It follows that
$\beta$ has multiplicity $m-1$ in $\gcd(H(\alpha,y),\frac{\partial H}{\partial y}(\alpha,y))$ and
$m$ in $\gcd(H(\alpha,y),(\frac{\partial H}{\partial y})^2(\alpha,y))$ because $m\leq 2m-2$, that is
$m-1\geq 1$, since
$\beta$ is solution of $\frac{\partial H}{\partial y}(\alpha,y)$.

We denote  the multiplicity of $\beta$ in $\gcd(P(\alpha,y),Q(\alpha,y))$  as $\text{mult}(\beta,\gcd(P(\alpha,y),Q(\alpha,y)))$.
Summing 
over all the critical points of $H$ and noticing that the set  $V_H$ of distinct solutions of $\{
H,\frac{\partial H}{\partial y}\}$ is the same as that of $\{ H,(\frac{\partial H}{\partial
y})^2\}$, we obtain that the number of critical points is
\[\begin{split}
\#V_H= &
 \sum_{(\alpha,\beta) \in V_H}{\text{mult}(\beta,\gcd(H(\alpha,y),(\frac{\partial H}{\partial y})^2(\alpha,y)))}\\
&-\sum_{(\alpha,\beta) \in V_H}{\text{mult}(\beta,\gcd(H(\alpha,y),\frac{\partial H}{\partial
y}(\alpha,y)))},
\end{split}\]
which is equal, by Lemma~\ref{lem:mult},  to the difference of the degrees of the decompositions of $\{ H,(\frac{\partial H}{\partial y})^2\}$ and $\{ H,\frac{\partial
  H}{\partial y}\}$. These degrees are  computed by Algorithm \ref{alg:tri-dec-mod2}, which
  concludes the proof of correctness of Algorithm~\ref{alg:rad-tri-dec}.

The complexity analysis of the algorithm directly follows from Lemma~\ref{lem:tridec_pair} noticing that
$\frac{\partial H}{\partial y}$ and $(\frac{\partial H}{\partial y})^2$ have degrees at most $2d$
(and bitsizes in $\OO(d+\tau)$ when defined over $\Z$) and that $(\frac{\partial H}{\partial y})^2$ can be computed from
$\frac{\partial H}{\partial y}$ in 
complexity $\sO(d^2)$ (and $\sOB(d^2\tau)$ when defined over $\Z$) \cite[Cor. 8.28]{vzGGer2}.
\end{proof}

\subsection{Lucky prime}

In Algorithm~\ref{algo:lucky}, we compute a lucky prime for $\{ H,\frac{\partial H}{\partial y} \}$ (see Definition
\ref{def:lucky-nb-roots}) in a straightforward manner by first computing the number of distinct
solutions of the system and then by computing the number of solutions
of its image modulo distinct prime numbers $\mu$ until the same number of solutions is found (and
checking that some leading coefficients do not vanish modulo $\mu$). Note that
Algorithm~\ref{algo:lucky} is a simplified variant of \blue{\cite[Algorithm 3]{bouzidiJSC2014a}} 
where we use here   the knowledge of the number of critical points of $H$
 to avoid  computing an explicit bound on
the number of unlucky primes.

\begin{algorithm}[t]
  \caption{Lucky prime for  $\{H,\frac{\partial H}{\partial y}\}$} 
\label{algo:lucky}
\begin{algorithmic}[1]

  \REQUIRE{ $H$ in $\Z[X,Y]$ such that $Lc_y(H) \in \Z$}

\ENSURE{A lucky prime $\mu$ for the system $\{H,\frac{\partial H}{\partial y}\}$}

\STATE $N$= Algorithm~\ref{alg:rad-tri-dec} ($H$) \label{line-distinct}

\STATE Compute $H(t-sy,y)$ and $ \frac{\partial H}{\partial y}(t-sy,y)\quad $
\label{alg:nbrootZ1}

\STATE $m=2d^4$

\WHILE {true} \label{alg:nbrootZ3}

\STATE Compute the set $B$ of the first $d^4+d^3\tau$ primes $>m$ 

\FORALL {$\mu$ in $B$}

\STATE Compute the reduction mod. $\mu$ of $H$, $\frac{\partial H}{\partial y}$, $L_H$,  $L_\frac{\partial
H}{\partial y}$\label{alg:nbrootZ6.1}

\IF{  $\phi_\mu(L_H(s))\ \phi_\mu(L_{\frac{\partial H}{\partial y}}(s)) \not\equiv 0$ \label{alg:nbrootZ4}} 
 \STATE Compute ${N}_\mu=$  Algorithm \ref{alg:rad-tri-dec}$(\phi_\mu(H),\phi_\mu(\frac{\partial H}{\partial y}))$\label{alg:nbrootZ5}
 \IF {${{N}}_\mu = N$\label{alg:nbrootZ6}} 
 \RETURN $\mu$
 \ENDIF
   \ENDIF
\ENDFOR
 
\STATE $m=$ the largest prime in $B$  

\ENDWHILE 

\end{algorithmic}
\end{algorithm}

\begin{proposition}\label{prop:comp-lucky}
Given $H\in\Z[x,y]$ of degree $d$ and bitsize $\tau$, 
Algorithm~\ref{algo:lucky} computes a lucky prime  for $\{H,\frac{\partial H}{\partial y} \}$ using $\sOB(d^7+d^6\tau)$ bit operations.
\end{proposition}
\begin{proof}
The correctness of Algorithm~\ref{algo:lucky} follows directly from the fact that the number of unlucky primes is finite 
(see~\blue{\cite[Prop. 13]{bouzidiJSC2014a}}). 

We now analyze the complexity of the algorithm. Computing the number of critical points of
$H$ in Line~\ref{line-distinct} has complexity $\sOB(d^7+d^6\tau)$ by
Proposition~\ref{prop:proof-correctness}.
It is
straightforward that the computations in Line~\ref{alg:nbrootZ1} can be done in bit complexity
$\sOB(d^4+d^3\tau)$ (see e.g. \blue{\cite[Lemma 7]{bouzidiJSC2014a}}). 
There are \blue{$\OO(\log d\tau)$} iterations of the loop in Line~\ref{alg:nbrootZ3} because there are
$\sO(d^4+d^3\tau)$ unlucky primes \blue{\cite[Prop. 13]{bouzidiJSC2014a}}. 
Each iteration of this loop consists in testing, for
the $d^4+d^3\tau$ 
primes in $B$, the non-vanishing of the reduction of the two polynomials
$L_H(s)$ and $L_{\frac{\partial H}{\partial y}}(s)$ and the equality between the number of
solution over $\Z$ and its analogue over $\Z_\mu$.

Polynomials $H$, $\frac{\partial H}{\partial y}$, $L_H$ and $L_\frac{\partial
H}{\partial y}$ are of degree at most $d$ in one or two variables and
they have bitsize at most $\sO(d+\tau)$ (see e.g. \blue{\cite[Lemma 7]{bouzidiJSC2014a}}).
The reduction of all their $O(d^2)$ coefficients modulo all the primes in $B$ can be computed via a remainder tree in a bit
  complexity that is soft linear in the total bitsize of the input
\cite[Thm. 1]{moenck1974}, which is dominated by the sum of the
    bitsizes of the $d^4+d^3\tau$ primes 
 in $B$ each being of bitsize
     \blue{$\OO(\log d\tau)$} (since there are \blue{$\OO(\log d\tau)$} iterations of the loop in Line~\ref{alg:nbrootZ3}). Hence, the bit complexity of  Line~\ref{alg:nbrootZ6.1} is~$\sOB(d^4+d^3\tau)$.

Finally, the arithmetic complexity of
     Algorithm \ref{alg:rad-tri-dec} is in $\sO(d^3)$, 
by Lemma~\ref{prop:proof-correctness},  thus its bit complexity is also in $\sOB(d^3)$
     since \blue{$\mu \in \OO(\log d\tau)$}.  Hence, the total bit complexity of Line~\ref{alg:nbrootZ5} is
     $\sOB(d^7+d^6\tau)$, and so is the bit complexity of one iteration of the loop in
     Line~\ref{alg:nbrootZ3}. Since at most \blue{$\OO(\log d\tau)$} iterations are performed, this yields an
     overall bit complexity for Algorithm \ref{algo:lucky} in $\sOB(d^7+d^6\tau)$.
\end{proof}

\section{Wrap up}
The results of the previous sections can easily be combined in the following theorem.

\begin{theorem}\label{thm:final}
Let $P,Q \in \Z[x,y]$ of total degree at most $d$ and maximum bitsize $\tau$. A separating linear form $x+ay$ for $\{P,Q\}$ with $a$ an integer of bitsize in \blue{$\OO(\log d)$}  can be computed using $\sOB(d^7+d^6\tau)$ bit operations.
\end{theorem}

\begin{proof}
According to Lemma~\ref{lem:system2curve}, we can compute in complexity $\sOB(d^6+d^5\tau)$ a
\blue{shearing} $t=x+\alpha y$, $\alpha\in O(d)$, and a polynomial $H \in \Z[t,y]$ of total degree at most
$2d$ and bitsize  $\sO(d+\tau)$ such that $\{H,\frac{\partial H}{\partial y}\}$ is
zero-dimensional and such that $x+(\alpha+a)y$ is separating for
$\{P,Q\}$ if $t+ay$ is separating for $\{H,\frac{\partial H}{\partial y}\}$. The result follows since by
Propositions~\ref{prop:sep-elem-comp},~\ref{prop:proof-correctness}
and~\ref{prop:comp-lucky}, since an
 integer $a$ in $\OO(d^4)$ such that $t+ay$ separates $\{H,\frac{\partial H}{\partial y}\}$ can be computed using  $\sOB(d^7+d^6\tau)$ bit operations.
\end{proof}

\section{\blue{Las-Vegas algorithm}} 
\label{sec:Las-Vegas}

\blue{In this section, we present a Las-Vegas version of the algorithm presented in the
previous sections, whose expected bit complexity is $\sOB(d^5+d^4\tau)$ 
(Theorem~\ref{thm:las-vegas}).}


The Las-Vegas version of our algorithm is the same as the deterministic one except that we use
Las-Vegas algorithms for gcd computations and that we choose randomly candidates for a separating linear
form and a lucky prime
 in Algorithms~\ref{alg:sep-elem} and \ref{algo:lucky}.

 More precisely, \blue{in the Las-Vegas version of}
 Algorithm~\ref{alg:sep-elem}, the separating linear form is computed by
 choosing 
 at random an integer $a$ in $[0,4d^4]$ until a candidate satisfying the
 condition of Line~\ref{alg:sep9} is found. There are at most $2d^4$ integers
 that do not satisfy this condition,\footnote{\small Indeed, $\Upsilon$ is of
   degree at most $2d$ and the system $\{P_\mu,Q_\mu\}$ has at most $d^2$
   solutions which define at most $d^2\choose 2$ directions in which two
   solutions are aligned. Furthermore $2d+{d^2\choose 2}< 2d^4$ (for $d\geq
   2$).}  thus a good candidate is chosen with probability at least $\frac12$,
 and so at most 2 candidates are chosen on average.

\blue{In the Las-Vegas version of} Algorithm~\ref{algo:lucky}, we  first compute a set $B$ of $2m$ prime numbers where $m$ is an
upper bound on the number of unlucky primes for $\{H,\frac{\partial H}{\partial y}\}$; such 
a set $B$ can be computed in bit complexity $\sOB(d^4+d^3\tau)$ (see the proof of~\blue{\cite[Lemma 18]{bouzidiJSC2014a}}).  
Then, we iteratively choose at random a prime number $\mu$ in $B$ until the conditions of
Algorithm~\ref{algo:lucky} are satisfied (Lines~\ref{alg:nbrootZ4} and \ref{alg:nbrootZ6}). 
The primes not satisfying these conditions are the unlucky ones, by definition, thus a lucky prime
is found with probability at least $\frac12$, and so at most 2 candidates are
chosen on average.

It remains to prove that the expected bit complexity of Algorithms~\ref{alg:sep-elem},
\ref{alg:rad-tri-dec}, and \ref{algo:lucky}, as well as the initial \blue{shearing} of
the coordinate systems, are in $\sOB(d^5+d^4\tau)$.
Our analysis is based on the following result on the expected complexity  of gcd computations.

\begin{lemma}[{\cite[Cor. 11.11]{vzGGer2}}]\label{lem:expected-gcd}
Let $f,g \in \Z[x]$ of degree at most $d$ and maximum bitsize $\tau$. The $\gcd$ of $f$ and $g$ can be computed using an expected number of $\sOB(d^2+d\tau)$ bit operations.
\end{lemma}

\begin{lemma}\label{lem:sep-elem-comp-exp}
Given  $P,Q \in \Z[x,y]$ of total degree at most $d$ and maximum bitsize $\tau$,
the Las-Vegas version of Algorithm~\ref{alg:sep-elem}  computes a separating linear form $x+ay$ for $\{P,Q\}$ with
$a<2d^4$ with an expected bit complexity in  $\sOB(d^4+d^3{\tau})$. 
\end{lemma}
\begin{proof}
  In the proof of Proposition~\ref{prop:sep-elem-comp}, we proved that the complexity of the algorithm
  is $\sOB(d^4+d^3\tau)$ plus $\sOB(d^3)$ times the number of considered choices of integer $a$. As
  argued above, at most two candidate integers are considered on average, which yields the lemma.
\end{proof}

\begin{lemma}\label{lem:lem:tridec_pair-exp}
Given  $P,Q \in \Z[x,y]$ of total degree at most $d$ and maximum bitsize $\tau$,
Algorithm~\ref{alg:tri-dec-mod2} computes the degree of the triangular decomposition of $\{P,Q\}$ 
with an expected bit complexity in $\sOB(d^5+d^4{\tau})$. 
\end{lemma}
\begin{proof}
According to Lemma~\ref{complexity:subresultant}, the sequence of the principal subresultant
coefficients $sres_{i,y}(P,Q)$, $i=0,\ldots,d$ can be computed in  $\sOB(d^4\tau)$ bit operations,
and each of these principal subresultants (including the resultant) has degree $\OO(d^2)$ and bitsize
$\sO(d\tau)$. The algorithm then performs  at most $d$ gcd computations between these polynomials
(including the computation of the squarefree part of the resultant). By Lemma~\ref{lem:expected-gcd}
and using Mignotte's bound~\cite[Cor. 10.12]{BPR06}, each of these gcds can be computed in  an
expected bit complexity $\sOB(d^4+d^3\tau)$. Hence computing $d$ such
gcds can be done with expected bit complexity $\sOB(d^5+d^4\tau)$, which concludes the proof.
\end{proof}

\begin{lemma}\label{lem:proof-correctness-exp}
Given  $P,Q \in \Z[x,y]$ of total degree at most $d$ and maximum bitsize $\tau$,
  Algorithm \ref{alg:rad-tri-dec} 
computes the number of critical points of $H$ 
with an expected bit complexity in $\sOB(d^5+d^4{\tau})$. 
\end{lemma}
\begin{proof}
As seen in the proof of Proposition~\ref{prop:proof-correctness}, Algorithm \ref{alg:rad-tri-dec}
computes $\frac{\partial H}{\partial y}$ and $(\frac{\partial H}{\partial y})^2$ in bit complexity
$\sOB(d^2\tau)$ and calls Algorithm~\ref{alg:tri-dec-mod2} on two systems of degrees $O(d)$ and
bitsizes $O(d+\tau)$. The result follows from Lemma~\ref{lem:lem:tridec_pair-exp}.
\end{proof}

\begin{lemma}\label{lem:lucky-exp}
Given $H \in \Z[x,y]$ of total degree $d$ and bitsize $\tau$, 
the Las-Vegas version of Algorithm~\ref{algo:lucky} computes a lucky prime  for $\{H,\frac{\partial H}{\partial y}\}$ with an expected bit complexity in  $\sOB(d^5+d^4\tau)$.
\end{lemma}
\begin{proof}
By Lemma~\ref{lem:proof-correctness-exp}, the first call to Algorithm \ref{alg:rad-tri-dec} has
an expected bit complexity in $\sOB(d^5+d^4{\tau})$. As shown in the proof of Proposition~\ref{prop:comp-lucky}, the
bit complexity of shearing $H$ and $\frac{\partial H}{\partial y}$, as well as the reductions modulo
$\mu$ have  bit complexity $\sOB(d^4+d^3{\tau})$. Finally, by Proposition~\ref{prop:proof-correctness}, the calls  in $\Z_\mu$ to Algorithm
\ref{alg:rad-tri-dec} have arithmetic complexity $\sO(d^3)$ and thus bit complexity $\sOB(d^3)$
(since $\mu\in \OO(\log d\tau)$). This concludes the proof since, as discussed above,  the expected number of such calls is
at most 2.
\end{proof}

Combining the above results, we obtain the following theorem.

 \begin{theorem}\label{thm:las-vegas}
Let $P,Q \in \Z[x,y]$ of total degree at most $d$ and maximum bitsize $\tau$. A separating linear form $x+ay$ for $\{P,Q\}$ with $a$ an integer of bitsize in $\OO(\log d)$  can be computed using an expected number of $\sOB(d^5+d^4\tau)$ bit operations.
\end{theorem}

\begin{proof}
As in the proof of Theorem~\ref{thm:final}, Lemmas~\ref{lem:sep-elem-comp-exp},
\ref{lem:proof-correctness-exp} and \ref{lem:lucky-exp} yield the result once we prove, as in
Lemma~\ref{lem:system2curve}, that we can compute with the right complexity a
\blue{shearing}
 $t=x+\alpha y$, $\alpha\in O(d)$, and a polynomial $H \in \Z[t,y]$ of total degree at most
$2d$ and bitsize  $\sO(d+\tau)$ such that $\{H,\frac{\partial H}{\partial y}\}$ is
zero-dimensional and such that, if $t+ay$ is separating for $\{H,\frac{\partial H}{\partial y}\}$, then 
$x+(\alpha+a)y$ is separating for $\{P,Q\}$.
 
According to the proof of Lemma~\ref{lem:system2curve}, we only need to prove that the computation of the squarefree part of $H=P\,Q$ can be done with an expected bit complexity in $\sOB(d^5+d^4\tau)$. Replacing in the proof of \cite[Lemma 13]{sagraloff2013} the bit complexity of computing a univariate gcd by the one in Lemma~\ref{lem:expected-gcd} yields the result.  
\end{proof}

\section{Conclusion}

This paper focuses on the computation of separating linear forms for bivariate systems.  First, we
proved that the computation of such a separating form can be done with a bit complexity
$\sOB(d^7+d^6\tau)$ in the worst case.  As mentioned in the introduction, this result directly
yields, within the same worst-case bit complexity, the rational parameterization of Gonzalez-Vega et
al. \cite{VegKah:curve2d:96,det-jsc-2009} and 
 that of Rouillier
\cite{Rou99,bouzidiJSC2014a}. 
  Second, we proved that the computation of a separating linear
form can be done in a Las-Vegas setting using an expected number of $\sOB(d^5+d^4\tau)$ bit
operations.  As a consequence, the computation in this setting of a separating linear form now becomes
non-dominant in the whole process of computing a rational parameterization; \blue{indeed, given a
separating linear form, computing Gonzalez-Vega et al. and 
Rouillier's parameterizations both have bit
complexity in $\sOB(d^7+d^6\tau)$ even in the Las-Vegas setting.

It should be mentioned that the best known upper bound for the total bitsize of the parameterization
of Gonzalez-Vega et al. is $\sO(d^5+d^4\tau)$.\footnote{\small Indeed, the approach of Gonzalez-Vega
  et al. first applies a linear change of variables to the input polynomials, which increases the
  bitsize of the polynomials to $\tau'\in\sO(d+\tau)$, and then computes rational parameterizations
  of the solutions of the $O(d)$ systems of the triangular decomposition
  (Algorithm~\ref{alg:tri-dec-mod}). The rational parameterizations are ratios of coefficients of
  the polynomial subresultants (seen as polynomials in $y$) which have degrees $O(d^2)$ and bitsize
  $\sO(d\tau')=\sO(d^2+d\tau)$ (Lemma~\ref{complexity:subresultant}). The total bitsize of the
  $O(d)$ parameterizations is thus $\sO(d^5+d^4\tau)$.}  Thus, some progress on this upper bound
would be required before any further progress on the computation of a separating linear form in the
Las-Vegas setting could impact that of computing this parameterization.  However, note that the
situation is slightly different for the Rational Univariate Representation (RUR) of Rouillier
\cite{Rou99} whose total bitsize is $\sO(d^4+d^3\tau)$~\cite[Theorem 22]{bouzidiJSC2014a}.

Finally, we note that, for computing a separating linear form of an
\emph{arbitrary} system $\{P,Q\}$,  the algorithm presented here is likely purely theoretical
because considering the system $\{PQ,\frac{\partial PQ}{\partial y}\}$ instead $\{P,Q\}$ essentially doubles the
degree of the input polynomials, which is likely not efficient in practice. However, for the problem
of computing the critical points of a curve, there is some
good hope that our  algorithm is efficient in practice.}

\small
\bibliographystyle{abbrv}
 \bibliography{bib-algcurves}

\end{document}